\documentclass[11pt]{article}
\usepackage{graphicx}
\usepackage{stmaryrd}
\usepackage{amssymb}
\usepackage{amsthm}
\usepackage{dsfont}
\usepackage{hyperref}
\usepackage{mathrsfs}
\usepackage{stmaryrd}
\usepackage[lite,initials,msc-links]{amsrefs}
\usepackage{amsmath}
\usepackage{centernot}
\usepackage{enumerate}
\usepackage{verbatim}
\usepackage{xcolor}
\usepackage[margin=1.2in]{geometry}

\newtheorem{theorem}{Theorem}[section]

\newtheorem{corollary}[theorem]{Corollary}
\numberwithin{equation}{section}

\theoremstyle{definition}

\theoremstyle{remark}
\newtheorem{remark}{Remark} 

\newcommand{\clust}{\mathscr{C}}
\newcommand{\coin}{\xi}
\newcommand{\field}{\mathcal{S}}
\newcommand{\IP}{\mathbf{P}}
\newcommand{\IE}{\mathbf{E}}
\newcommand{\graph}{G}

\newcommand{\set}{\omega}
\newcommand{\sets}{\Omega}
\newcommand{\currs}{\Omega}
\newcommand{\flows}{\mathcal{F}}
\newcommand{\rcur}{\textnormal{curr}}
\newcommand{\drcur}{\textnormal{d-curr}}
\newcommand{\aflow}{\textnormal{a-flow}}

\newcommand{\match}{M}
\newcommand{\dimer}{\textnormal{dim}}

\newcommand{\cont}{C}
\newcommand{\matchs}{\mathcal{M}}

\title{On the double random current nesting field}

\author{Hugo Duminil-Copin\thanks{Institut des Hautes \'Etudes Scientifiques} \thanks{Universit\'e de Gen\`eve}\qquad \qquad Marcin Lis\thanks{University of Vienna}}

\date{\today}

\begin{document}
\maketitle

\begin{abstract}
We relate the planar random current representation introduced by Griffiths, Hurst and Sherman to the dimer model. More precisely, 
we provide a measure-preserving map between double random currents (obtained as the sum of two independent random currents) on a planar graph and dimers on an associated bipartite graph. We also define a nesting field for the double random current, which,
under this map, corresponds to the height function of the dimer model. 
As applications, we provide an alternative derivation of some of the bozonization rules obtained recently by Dub\'edat, and show that the spontaneous magnetization of the Ising model on a planar biperiodic graph vanishes
at criticality. 

\end{abstract}

\section{Introduction}

The goal of this paper is to present a new connection between the Ising model and dimers through double random currents, and to show some of its applications. The link between dimers and the Ising model has a long history that we will not describe in detail here (we refer the reader to the extensive literature for more information). The articles that we choose to mention in the introduction are the ones directly relevant to our new connection.

\subsection{Random currents and dimers}

The Ising model is a random configuration of $\pm 1$ spins. In this article we think of the spins as living on the faces of a planar graph $G=(V,E)$ with vertex set $V$ and edge set $E$. 
In  \cite{Pei36} Peierls used the so-called {\em low-temperature expansion} of the model to show the existence of an order-disorder phase transition in the Ising model on $\mathbb Z^2$. In this representation, configurations of spins assigned to the faces of $G$ are mapped to contour configurations on $G$. More precisely, for $B\subset V$, write $\mathcal{E}^B$ for the collection of sets of edges $\set\subseteq E$ such that the graph $(V,\set)$ has odd degrees at $B$ and
even degrees everywhere else. A connected component of $\set\in \mathcal{E}^B$ is called a \emph{contour}, and $\set$ itself is called a \emph{contour configuration}.
Each spin configuration on the faces of $G$ is naturally associated with the collection $\omega$ of edges bordering two faces with different spins. Clearly, $\omega$ belongs to $\mathcal E^\emptyset$, and conversely, every element of $\mathcal E^\emptyset$ is associated with exactly two spin configurations, one with spin $+1$ on the unbounded face, and one with spin~$-1$.

The low-temperature expansion is only one among many classical representations of the Ising model. A few years after Peierls, van der Waerden \cite{Wae41} introduced the {\em high-temperature
 expansion}, which was also fruitfully used to study the Ising model on arbitrary graphs. In \cite{GHS} Griffiths, Hurst and Sherman proposed to expand the partition function (or more complicated weighted sums) 
 of the Ising model into a power series in the inverse temperature and expressed it in terms of integer-valued functions on the edges of~$G$. This new method, later called the {\em random current representation}, is particularly useful when studying truncated spin correlations and has since then been a central tool in the study of the Ising model.
  
  \begin{figure}
		\begin{center}
			\includegraphics[scale=0.4]{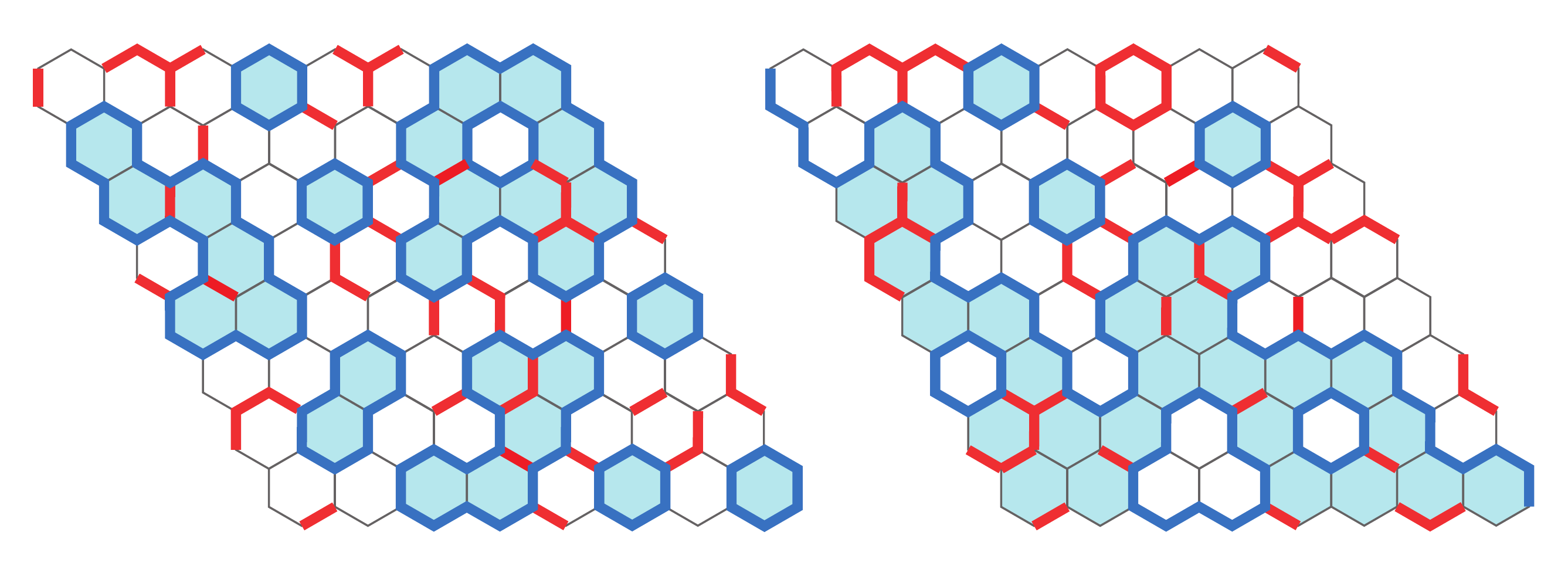}  
		\end{center}
		\caption{Two random current configurations on a piece of the hexagonal lattice with $\emptyset$, and $\{a,b\}$-boundary 
		conditions\ from left to right respectively, where $a$ is the top-leftmost and $b$ the bottom-rightmost vertex. 
		The odd edges are drawn in blue and the even edges in red.
		The colored faces are assigned spin $-1$ in the contour (low-temperature) expansion of the Ising model with $+$ and Dobrushin boundary conditions respectively.
		The interpretation of the odd edges of a current as contours of the Ising model is a consequence of the Kramers--Wannier duality~\cite{KraWan1}.}
	\label{fig:currents}
\end{figure}

In this article a \emph{current} on $G$ with {\em sources} $B\subset V$ is a set of edges $\set\subseteq E$ partitioned into two distinguished subsets
$\set_{\rm odd}\subseteq\set$ and $\set_{\rm even} \subseteq \set$, called \emph{odd} and \emph{even edges} respectively, such that $\set_{\rm odd}\in \mathcal{E}^B$
and $\set_{\rm even} = \set \setminus \set_{\rm odd}$. 
The set of all currents with sources $B$ will be denoted by $\sets^B$.
Let us also introduce the following probability measure on currents with sources $B$.
For each $e\in E$, fix $x_e \in [0,1]$ and set $p_e =1- \sqrt{1-x_e^2}$.
The random current model with sources $B$ is a probability measure
on $\sets^B$ given by
\begin{align}\label{eq:rc}
{\IP}^B_{ \rcur}(\set) = \frac 1{{Z}^B_{ \rcur}} \prod_{e\in \set_{\rm odd}} x_e \prod_{e \in \set_{\rm even}}p_e  \prod_{e\in E\setminus\set} (1-p_e), \qquad \text{for all }\set \in \sets^B,
\end{align}
where ${Z}^B_{ \rcur}$ is the partition function. 

\begin{remark}
Our definition of random currents is derived directly from the original one of Griffiths, Hurst and Sherman~\cite{GHS}, where a current is a function assigning to each edge a natural number.
It is left to the reader to check that our representation is obtained by forgetting the numerical value of the current but keeping the information about its parity and whether it is zero or not. More precisely, $\omega_{\rm odd}$ is the set of edges with odd current, $\omega_{\rm even}$ with strictly positive even current, and $E\setminus\omega$ with zero current.
\end{remark}

The random current model has been successful in several ways.  In the original article \cite{GHS}, it was used to derive correlation inequalities. In 1982 it was used by Aizenman \cite{Aiz82} to prove triviality of the Ising model in dimension $d\ge5$ and a few years later, Aizenman, Barsky and Fernandez proved that the phase transition is sharp \cite{AizBarFer87} (see also \cite{DumTas15} for an alternative proof). In recent years the representation has been the object of a revived interest. It was used to study the continuity of the phase transition (see below) and it was also related to other models. For instance,  a new distributional
relation between random currents, Bernoulli percolation and the FK-Ising model was discovered by Lupu and Werner \cite{LupWer}. 
For a more exhaustive account of random currents, we refer the reader to \cite{Dum16}.

In most applications, one considers pairs of independent current configurations. The reason comes from the combinatorial properties that this ``double current'' model enjoys. 
For two currents $\set$ and $\set'$, define the sum $\set+\set'$ to be the current with odd edges $\set_{\rm odd}\triangle \set'_{\rm odd}$ and even edges $(\set\cup \set') \setminus (\set_{\rm odd}\triangle \set'_{\rm odd})$, where 
$\triangle $ is the symmetric difference. 
This simply corresponds to addition mod 2 together with keeping track of whether the current is zero or not. 
Note that if $\set\in \sets^B$ and $\set' \in \sets^\emptyset$, then $\set+\set' \in \sets^B$.
Define the {\em double random current model} with sources~$B$ to be the probability measure on $\sets^B$ induced by the sum of two independent random currents
with sources\ $B$ and $\emptyset$:
\[
\IP^{B}_{\drcur}(\set) =\IP^B_{\rcur} \otimes \IP^{\emptyset}_{\rcur}(\{(\set',\set'') \in \sets^B \times \sets^\emptyset: \set' + \set'' = \set \}), \quad \text{for all }\set \in \sets^B.
\]

In \cite{LisT} the double random current model was represented in terms of so-called alternating flows studied by Talaska~\cite{talaska} in relation to the totally positive Grassmannian~\cite{postnikov}.
In this paper, inspired by the connection of another classical model of statistical physics, namely the dimer model, and the totally positive Grassmannian~\cite{PosSpeWil, Lam,LamNotes}, we relate
the double random current model to the dimer model. It turns out that our approach is also closely related to the correspondence between the double Ising model and the dimer model obtained by 
Dub\'edat \cite{Dub} (see Sec.\ \ref{sec:bozonization}).
Formally, the \emph{dimer model} is a probability measure on \emph{dimer covers} (also called \emph{perfect matchings}) of a graph, i.e.~sets of edges such that each vertex is incident on exactly one edge.  
We will now define a weighted graph $G^d$ on which the dimer model will be in a correspondence with double random currents. To this end, we proceed in two steps. We first define a directed graph $\vec\graph$ and then construct $G^d$ from it.

Let $\vec \graph$ be a directed graph with the same vertex set $V$ as $G$, and with edge set $\vec E$ defined as follows: 
each $e\in E$ is replaced by three parallel directed edges with the same endpoints as $e$, and such that that the middle edge $\vec e_m$ 
has the opposite orientation to the two side edges $\vec e_{s1}$ and $\vec e_{s2}$, see Fig.~\ref{fig:graphs}. The middle edge can be oriented arbitrarily, and
it is assigned weight $x_{\vec e_m} = \tfrac{2x_e}{1-x_e^2}$, whereas the side edges get weights $x_{\vec e_{s1}}=x_{\vec e_{s2}}=x_e$. 

\begin{figure}
		\begin{center}
			\includegraphics[scale=0.9]{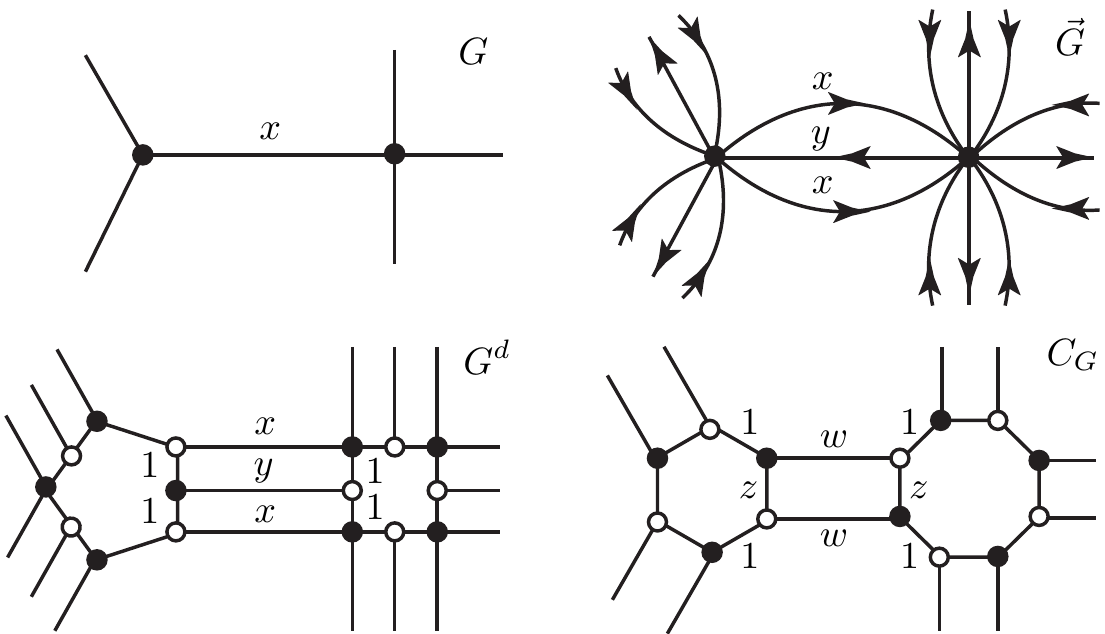}  
		\end{center}
		\caption{An example of the local structure of the graphs $G$, $\vec G$, $G^d$ and $C_G$. The weights satisfy $y=\tfrac{2x}{1-x^2}$, $w=\tfrac{2x}{1+x^2}$, $z=\tfrac{1-x^2}{1+x^2}$
		}
	\label{fig:graphs}
\end{figure}

The graph $G^d$ is constructed from $\vec G$ as follows (the reader may look at Fig.~\ref{fig:graphs} for an illustration). For a vertex $z$, let $r(z)$ be the number of pairs of consecutive edges in $\vec E$ around $z$ with the same orientation, and let $\text{deg}(z)$ be the degree of $z$.
Replace each $z$ with a cycle of $3\text{deg}(z)-r(z)$ edges, called \emph{short edges}. By construction, the length of the cycle is even, and hence  its vertices can be colored
black and white in an alternating way. Now, add \emph{long edges} corresponding to the edges of $\vec G$. We do it in such a way that if $(z,w)$ is a directed edge of $\vec G$, then
the corresponding edge in $G^d$ connects a white vertex in the cycle of $z$ with a black vertex in the cycle of $w$, and moreover, the cyclic order of edges around each cycle in $G^d$ 
matches the one in $\vec G$. The resulting graph $G^d$ is therefore \emph{bipartite}. We finish the construction by assigning weights.
The long edges inherit their weights from their counterparts in $\vec G$, and short edges get weight $1$. 

Let $\matchs^\emptyset$ be the set of dimer covers of $G^d$.
Define the dimer model probability measure with $\emptyset$ boundary conditions by
\begin{align}
\label{eq:dimer}
\IP^{\emptyset}_{\dimer} (\match) = \frac1{Z^{\emptyset}_{\dimer}}  \prod_{ e\in \matchs^\emptyset} x_{ e}, \qquad\text{for all } \match \in \matchs^\emptyset.
\end{align}

Let us now describe a mapping $\pi$ from the dimer covers of $G^d$ to current configurations on $G$. Consider a dimer cover $M$, and set $\pi(M)$ to be the current configuration $\omega\in\currs^\emptyset$ defined as follows: an edge $e$ of $G$ will be in $\omega_{\rm odd}$ (resp.~$\omega_{\rm even}$ and $E\setminus \omega$) if there is 1 or 3 dimers (resp.~2 and 0)  covering the three edges of $G^d$ associated with $e$ (see Fig.\ \ref{fig:currflowdimer}). One can check that the image of this map is included in $\Omega^{\emptyset}$, i.e., that the map always yields a sourceless current configuration.
 Let $\pi_*\IP_{\dimer}^\emptyset$ be the pushforward measure on $\currs^\emptyset$.
The main result of this paper is the following.
\begin{theorem} \label{thm:dimercurr}
For any finite simple planar graph $G$, we have $\pi_* \IP_{\dimer}^\emptyset=\IP_{\drcur}^{\emptyset}$.
\end{theorem}

\begin{remark}
The theorem can be extended to graphs that are properly embedded in an orientable surface. 
\end{remark}

\begin{figure}
		\begin{center}
			\includegraphics[scale=0.9]{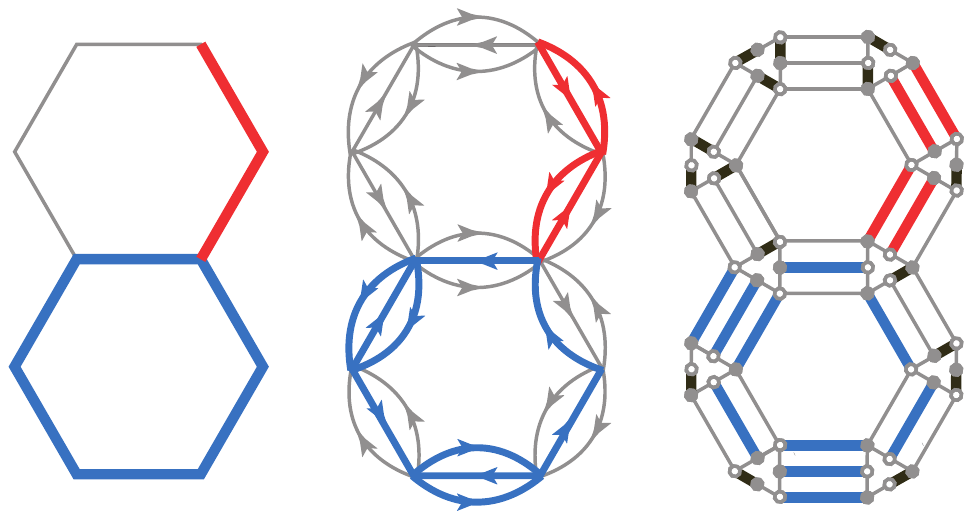}  
		\end{center}
		\caption[]{Left: A current configuration $\omega$ on $G$ with odd edges marked blue and even edges marked red. Center: 
		An alternating flow $F$ on $\vec G$ corresponding to $\omega$, i.e., such that $\theta(F)=\omega$, where $\theta$ is the map from Theorem~\ref{thm:firstmapping}. Right: A dimer cover $M$ on $G^d$
		associated with $F$ and $\omega$, i.e., such that $\eta(M)=F$ and $\pi(M)=\theta \circ \eta (M)=\omega$, where $\eta$ and $\pi$ are as in Theorem~\ref{thm:secondmapping} and Theorem~\ref{thm:dimercurr} respectively.
		
\indent Both $\eta$ and $\theta$ are many-to-one maps. In the example above, $|\theta^{-1}[\omega]|=2\times27$. The different possible orientations of the outer boundary of the flow account for the factor $2$ (see Fig.~\ref{fig:flowcurrs}), and every second odd edge of the cycle can be represented in exactly three ways, independently. Also, $|\eta^{-1}[F]|=4$ since each of the cycles of short edges in $G^d$ corresponding 
to an isolated vertex of $\omega$ can be covered by dimers in two ways, independently}
	\label{fig:currflowdimer}
\end{figure}

\subsection{The nesting field of a double random current} \label{sec:nestingfield} 
One of the main applications of Theorem~\ref{thm:dimercurr} is the study of the so-called {\em nesting field}. The idea behind introducing the nesting field is the interpretation of the contours of a current as level lines of a random surface whose discretization
is an integer-valued function defined on the faces of $G$.
The change in height of the discretized surface
when crossing a contour is either $+ 1$ or $-1$, and for two contours belonging to different clusters, the respective height changes are independent.

For a current $\set$, a connected component of the graph $(V,\set)$ will be called a \emph{cluster}. In particular, each contour $\cont$ of $\set_{\rm odd}$ (also called a contour of $\set$) is contained in a unique cluster of $\set$, and each
cluster $\clust$ of $\set$ gives rise to a contour configuration $\clust \cap \set_{\rm odd}$.
Call a cluster $\clust$ \emph{odd around a face} $u$ if the spin configuration associated via the low-temperature expansion with the contour configuration $\clust \cap \set_{\rm odd}$ assigns spin $-1$ to $u$ if the exterior face has spin $+1$. 

Let $(\coin_{\clust})$ be a family (indexed by clusters of $\omega$) of iid random variables equal to $+1$ or $-1$ with probability $1/2$. The \emph{nesting field} at $u$ is defined by
\[
\field_u = \sum_{\clust \text{ odd around } u }  \coin_{\clust},
\]
where the sum is taken over all clusters that are odd around $u$.  

One of the main features of Theorem~\ref{thm:dimercurr} is that it enables to connect the nesting field of a random current $\set$ drawn from the double random current measure to the height function associated with dimer covers of $G^d$. 
While the latter notion is classical, we still take a moment to recall it here. In the whole article, a {\em path} is a sequence of neighboring faces. 

To each dimer cover $M$ on $G^d$, we associate a \emph{$1$-form} $f_M$ (i.e.\ a function defined on directed edges which is antisymmetric under changing orientation) satisfying 
$f_M((z,w))=-f_M((w,z))=1$ if $\{z,w\} \in M$ and $z$ is white, and $f_M((z,w))=0$ otherwise. From now on, we fix a \emph{reference 1-form} $f_0$ 
given by $f_0((z,w))=-f_0((w,z))=1/2$ if $\{z,w\}$ is a short edge and $z$ is white, and $f_0((z,w))=0$ otherwise.

The \emph{height function} $h=h_M$ of a perfect matching $M$ is defined on the faces of $G^d$ and is given by
\begin{enumerate}[(i)]
\item $h(u_0)=0$ for the unbounded face $u_0$, 
\item for every other face $u$, choose a path $\gamma$ connecting $u_0$ and $u$, and define $h(u)$ to be the total flux of $f_M-f_{0}$ 
through $\gamma$, i.e., the sum of values of $f_M-f_{0}$ over the edges crossing $\gamma$ from left to right.
\end{enumerate}
The height function is well defined, i.e.\ independent of the choice of $\gamma$, since $f_M-f_{0}$ is a divergence-free flow.

Note that both the faces and vertices of $G$ are embedded naturally in the faces of $G^d$.
\begin{theorem} \label{thm:nestingfield}
The law of $h$ under $ \IP_{\dimer}^\emptyset$ restricted to the faces of $G$ is the same as the law of the nesting field $\field$ under $\IP_{\drcur}^{\emptyset}$.
\end{theorem}

\begin{remark} \label{rem:torusfield}
Again, the theorem can be extended to graphs $G$ that are properly embedded in the torus. In this case, the total increment of the nesting field on $G$ between two faces $u$ and $v$, as it is the case for the dimer height function on $G^d$, is defined only up to homotopy of the path $\gamma$ connecting $u$ and $v$ along which the divergence free flows is summed up. We denote these increments by $\field_{\gamma}$ and $h_{\gamma}$ respectively, and conclude that $\field_{\gamma}$ drawn according to $\IP^{\emptyset}_{\drcur, G}$ has the same distribution as
$h_{\gamma}$ drawn according to $\IP^{\emptyset}_{\dim, G^d}$.
Also, after fixing $\gamma$, the increment $\field_{\gamma}$ is equal to the sum of the $\pm 1$ variables $\coin_{\clust}$ for the clusters $\clust$ that are {\em odd with respect to $\gamma$}, meaning that the contour configuration $\clust \cap \set_{\rm odd}$ crosses an odd number of edges of $\gamma$. 
\end{remark}

Consider an infinite {\em biperiodic} (i.e.~invariant under the action of a $\mathbb{Z}^2$-isomorphic lattice) planar graph $\mathbb G$. The graph $\mathbb G$ is assumed to be {\em nondegenerate}, in the sense that the complement of the edges is the union of topological disks (in other words, the faces are topological disks). Then, the dimer graph $\mathbb G^d$ constructed
as in the finite case, is biperiodic and bipartite.
The height function of dimers on biperiodic bipartite graphs has been studied in detail, for instance in \cite{KenOkoShe06}. Kenyon, Okounkov and Sheffield identified three possible behaviors depending on the phase: gaseous, liquid or frozen, in which the associated dimer model lies. In particular, the height function of dimers in the liquid phase, which is specified by the property that the characteristic polynomial has  zeroes on the torus $\mathbb T^2$, has unbounded fluctuations.
Let $\mathbb G_n = \mathbb G / (n\mathbb Z\oplus n \mathbb Z)$. 
The relation between the nesting field and the height function of dimers can be hence combined with Theorem~4.5 of \cite{KenOkoShe06} to give the following.

\begin{corollary} \label{cor:4}
Assume that the characteristic polynomial of the dimer model on $\mathbb G^d$ has a real zero on the torus $\mathbb T^2$, then 
\[
\lim_{n\to \infty} \IE_{\drcur, \mathbb G_n}^{\emptyset}[\mathcal S_{\gamma}^2]=\tfrac{1}\pi\log[|\phi(u)-\phi(v)|] +o(\log[|\phi(u)-\phi(v)|]),
\]
where the limit is taken for a fixed path $\gamma$ connecting $u$ and $v$, $\IE^{\emptyset}_{\drcur, \mathbb G_n}$ is the expectation with respect to $\IP^{\emptyset}_{\drcur}$ on $\mathbb G_n$, 
and $\phi$ is a linear bijection from $\mathbb R^2$ to $\mathbb R^2$.
\end{corollary}
We will not use the specific form of $\phi$, but let us say that it is expressed in terms of the characteristic polynomial.

\subsection{Application 1: Bozonization rules for the Ising model} \label{sec:bozonization}For a finite planar graph $G=(V,E)$, define the set $\Sigma_G$ of configurations $\sigma$ assigning to each vertex $u\in V$ a spin $\sigma_u$, equal to $+1$ or $-1$. The distribution of the Ising model with free boundary conditions on $G$ at inverse temperature $\beta$ and with coupling constants $(J_e)_{e\in E}$ is defined on $\Sigma_G$ by $$\mu_{G,\beta}^{\rm f}(\sigma)=\frac{1}{Z_{\rm Ising}}\exp \Big (-\beta {\bf H}_G(\sigma)  \Big)\qquad\text{for all } \sigma\in\Sigma_G,$$
where 
${\bf H}_G=-\sum_{\{u,v\} \in E} J_{\{u,v\}}\,\sigma_u\sigma_v$ is the Hamiltonian of the model.
 
By construction, the Ising model  is related to the double random current on $G$ with parameters $x_e=\tanh(\beta J_e)$ and hence, Theorem~\ref{thm:dimercurr} gives a connection between the Ising model and dimers on a bipartite graph. It is known since \cite{Fis66} that the Ising model on a graph $G$ is related to a dimer model on a modified graph, called the Fisher graph of $G$. This connection enables to express the partition function of the former model in terms of the partition function of the later, which is more amenable to computations. The Fisher graph of $G$ is not bipartite, a fact which renders the study of the dimer model on it more difficult. 

Recently, Dub\'edat \cite{Dub} (see also \cite{BoudeT}) proved that the Ising model can be related to a dimer model on a bipartite graph $C_G$ where each edge of $G$ is replaced by a quadrilateral and each vertex of degree $d$ by a $2d$-gon face (see Fig.~\ref{fig:graphs}). The dimer model defined in this article on $G^d$ can in fact be mapped to the dimer model on $C_G$ with weights as in Fig.~\ref{fig:graphs} via an explicit sequence of vertex splittings and urban renewals (operations which partially preserve the distribution of dimers, and in particular, the height function, see Remark~\ref{rmk:1}). This means that Dub\'edat's mapping and our mapping are two facets of the same relation.

In \cite{Dub}, Dub\'edat derived powerful bozonization rules expressing the square of averages of order and disorder variables in terms of averages of certain observables of the height function of a dimer model. Here, we provide an alternative proof of some of these relations (Lemma 3 of~\cite{Dub}). Before stating the result, we define the notion of a disorder variable. A {\em disorder line} $\ell$ is a continuous curve drawn in the plane in such a way that it avoids $V$ and crosses $E$ finitely many times. The {\em disorder variable} $\mu_\ell$ associated with $\ell$  corresponds to the change of the Hamiltonian
flipping the coupling $ J_e$ to $-J_e$ for edges $ e \in E$ which are traversed an odd number of times by  $\ell$.
Correspondingly, the correlation function involving a collection of disorder variables $(\mu_{\ell_j})_{1\leq j\le n}$ and a function $F:\Sigma_G\rightarrow\mathbb C$  is defined by
\begin{equation}  \label{eq:def_tau}
\mu_{G,\beta}^{\rm f}\Big[  F  \prod_{j=1}^n \mu_{ \ell_j} \Big]
:= \mu_{G,\beta}^{\rm f}\Big[  F \exp(-\beta\sum_{e\in E_{\rm odd}}2J_e\sigma_x\sigma_y)\Big]
,
\end{equation}
where $E_{\rm odd}$ is the set of edges $e\in E$ crossed an odd number of times by $\cup_{j=1}^n\ell_j$.
Recall that the faces and vertices of $G$ are embedded naturally in the faces of $G^d$, and hence, with a slight abuse of notation, we can speak of the height function evaluated at a vertex or a face of $G$.
\begin{theorem}\label{thm:bozonization}
Consider a finite planar graph $G$, and the dimer model on $G^d$ with the associated weights. For any vertices $x_1,\dots,x_k$ and any disordered lines $\ell_1,\dots,\ell_n$ starting from the unbounded face $u_0$ and ending in the faces $u_1,\dots,u_n$ respectively, we have that 
\begin{equation}
\mu_{G,\beta}^{\rm f}\Big[\prod_{i=1}^k\sigma_{x_i}\times \prod_{j=1}^n\mu_{\ell_j}\Big]^2={\bf E}_{\rm dim}^\emptyset\Big[\prod_{i=1}^k\sin(\pi h_{x_i})\times\prod_{j=1}^n\cos(\pi h_{u_j})\Big].
\end{equation}
\end{theorem}
Note that for a vertex $x$ and a face $u$, $\sin(\pi h_x)=(-1)^{h_x-1/2}$ and $\cos(\pi h_u)=(-1)^{h_u}$.
Also, the fact that the disorder lines are starting on the unbounded face $u_0$ is a convenient convention to state the result elegantly in terms of the notation introduced in the previous section. The theorem can be extended to graphs $G$ properly embedded in the torus with appropriate modifications.

\subsection{Application 2: Continuity of the phase transition for the Ising model on biperiodic planar graphs} For a finite subgraph $G=(V,E)$ of a nondegenerate biperiodic graph $\mathbb G=(\mathbb V,\mathbb E)$, define the set $\Sigma_G^+$ of configurations $\sigma$ assigning to each vertex of $\mathbb G$ a spin $\sigma_u$, equal to $+1$ or $-1$, with the additional constraint that any vertex of $\mathbb V\setminus V$ receives a spin $+1$. The distribution of the Ising model with $+$ boundary conditions on $G$ at inverse-temperature $\beta$ and with coupling constants $(J_e)_{e\in E}$ is defined on $\Sigma_G^+$ by
$$\mu_{G,\beta}^+(\sigma)=\frac{1}{Z_{\rm Ising}}\exp \Big (-\beta {\bf H}^+_G(\sigma) \Big)\qquad\text{for all } \sigma\in\Sigma_G,$$ 
with
${\bf H}_G^+:=-\sum_{\{u,v\}} J_{\{u,v\}}\,\sigma_u\sigma_v$, where the sum is over edges $\{u,v\}$ intersecting $V$.
A measure $\mu_{\mathbb G,\beta}^+$ can be defined on $\mathbb G$ by taking the weak limit of the measures $\mu_{G,\beta}^+$.
The model undergoes an order/disorder phase transition on $\mathbb G$ at a critical inverse-temperature $\beta_c=\beta_c(\mathbb G)$ characterized by  the property that 
$\mu_{\mathbb G,\beta}^+[\sigma_u]=0$ if $\beta<\beta_c$ and $\mu_{\mathbb G,\beta}^+[\sigma_u]>0$ if $\beta>\beta_c$, where $u$ is an arbitrary vertex of $\mathbb G$.

In \cite{CimDum13}, the critical parameter $\beta_c$ of the Ising model was proved to correspond to the only value of $\beta$ for which the dimer model introduced in \cite{Dub} on $C_{\mathbb G}$ (and therefore the one defined here on $\mathbb G^d$) is in the liquid phase. Here, we combined this result with the information above to prove the following statement.
\begin{theorem}\label{thm:continuity}
Let $\mathbb G$ be a nondegenerate infinite biperiodic planar graph, then $$\mu_{\mathbb G,\beta_c}^+[\sigma_u]=0.$$\end{theorem}

For the square lattice, the result goes back to the exact computation of Yang \cite{Yan52}. In higher dimension, the fact that $\mu_{\mathbb G,\beta_c}^+[\sigma_u]=0$ is known for the nearest neighbor Ising model on $\mathbb G=\mathbb Z^d$ \cite{AizDumSid15,AizFer86}. On trees, the result was proved in \cite{Hag96}. Recently, Raoufi \cite{Rao16} showed that amenable groups with exponential growth undergo a continuous phase transition.
To the best of our knowledge, a proof which is valid for any infinite biperiodic planar graph was not available until now. 

A byproduct of the proof is the following result about non-percolation of spins.
\begin{corollary}\label{cor:5}
Let $\mathbb G$ be a nondegenerate infinite biperiodic planar graph, then the $\mu_{\mathbb G,\beta_c}^+$-probability that there exists an infinite cluster of pluses or minuses is zero.
\end{corollary}

\subsection{Extension to Dobrushin boundary conditions}
Much of what has been described above can be extended to cover the case of the Ising model with Dobrushin boundary conditions. Consider  two vertices $a$ and $b$ on the exterior face of $G$. 
Configurations in $\mathcal{E}^{\{a,b\}}$ correspond to (the so-called Dobrushin) spin configurations where the external face is split into two faces of opposite spins by adding
an additional edge joining $a$ and $b$.
In particular, this construction implies that $\set \in \mathcal{E}^{\{a,b\}} $ necessarily contains a contour 
connecting $a$ and $b$.

The definition of the nesting field for a current with $\{a,b\}$-boundary conditions\ is almost the same with the exception that the variable $\coin_{\clust_0}$
corresponding to the cluster $\clust_0$ connecting $a$ and $b$ is set to $1$. Moreover, the cluster $\clust_0$ is called odd around $u$
if its contours assign spin $-1$ to $u$ in the model with Dobrushin boundary conditions with $+1$ spin on the (external) face adjacent to the clockwise
boundary arc from $a$ to $b$, and $-1$ spin on the face adjacent to the arc from $b$ to $a$.

Consider an augmented graph $\vec \graph_{(a,b)}$ where an additional edge $e_{(a,b)}$ 
directed from $b$ to $a$ is added in the external face of $\vec G$ in such a way that the clockwise boundary arc of $\vec G$ from $a$ to $b$ is bordering the unbounded face of $\vec \graph_{(a,b)}$. We define the graph $G^d_{(a,b)}$ out of $\vec G_{(a,b)}$ exactly as we defined $G^d$ out of $\vec G$.
Let $\matchs_{(a,b)}$ be the set of dimer covers of $G^d_{(a,b)}$ containing the edge $(b,a)$. Also, introduce the height function $h$ of $M$ in $\matchs_{(a,b)}$ by choosing a reference 1-form corresponding to a matching that represents a current composed only of a path of odd edges that form the clockwise arc from $b$ to $a$ on the boundary of $G$.

Then, we have the following extension of Theorems~\ref{thm:dimercurr} and \ref{thm:nestingfield}. 
\begin{theorem} \label{thm:dimercurr2}
For any finite simple planar graph $G$, 

(i) $\pi_* \IP_{\dimer}^{\{a,b\}}=\IP_{\drcur}^{\{a,b\}}$,

(ii) the law of $h$ under $ \IP_{\dimer}^{(a,b)}$ restricted to the faces of $G$ is the same as the law of the nesting field $\field$ under $\IP_{\drcur}^{\{a,b\}}$.

\end{theorem}

\bigbreak\noindent{\bf Acknowledgments}  M.L.\ is grateful to Sanjay Ramassamy for an inspiring remark, and the authors thank Aran Raoufi and Gourab Ray for many useful discussions. This article was finished during a stay of M.L. at IHES. M.L.\ thanks IHES for the hospitality.  The research of M.L.\ was funded by EPSRC grants EP/I03372X/1 and EP/L018896/1 and was conducted when M.L.\ was at the University of Cambridge.
The research of H.D.-C.\ was funded by a IDEX Chair from Paris Saclay, the ERC CriBLaM, and by the NCCR SwissMap from the Swiss NSF.

\section{Proofs of Theorems~\ref{thm:dimercurr}, \ref{thm:nestingfield} and \ref{thm:dimercurr2}}

There will be no difference in working with $B=\emptyset$ or $B=\{a,b\}$. For this reason, we simply refer to $B$ as being the set of sources. The proofs rely on the notion of \emph{alternating flows} and their \emph{height function}. For this reason, we define a probability measure on flows which will be later naturally related to the
double random current measure and its nesting field. We should mention that the proofs of the theorems can be obtained by hand, meaning without using alternating flows. Nonetheless, we believe that alternative flows offer an elegant way of deriving the connection between dimers and double random currents.

A \emph{sourceless alternating flow} $F$ is a set of directed edges of $\vec G$
such that for each vertex $v$, the edges in $F$ around $v$ alternate between being oriented towards and away from $v$ when going around $v$ (see Fig.~\ref{fig:flowcurrs}). In particular, the same number 
of edges enters and leaves~$v$. For two vertices $a$ and $b$ on the outer face of $\vec \graph$, an \emph{alternating flow with source $a$ and sink $b$} is a sourceless alternating flow on $\vec \graph_{(a,b)}$ containing $e_{(a,b)}$  (note that, here, $(a,b)$ is an oriented edge and should not be confused with $\{a,b\}$).
Denote the set of sourceless alternating flows on $\vec \graph$ by $\flows^\emptyset$, and the set of alternating flows with source $a$ and sink $b$
by $\flows^{(a,b)}$.

Define a probability measure on alternating flows with $B=\emptyset$ or $B=(a,b)$ by
\begin{align}
\label{eq:altflow}
\IP^{B}_{\aflow} (F) = \frac1{Z^{B}_{\aflow}} 2^{|V^c(F)|} \prod_{\vec e\in F} x_{\vec e}, \qquad \text{for all }F  \in \flows^B,
\end{align}
where $V^c(F)$ is the set of isolated vertices in the graph (V,$F$).

\begin{figure}
		\begin{center}
			\includegraphics[scale=0.9]{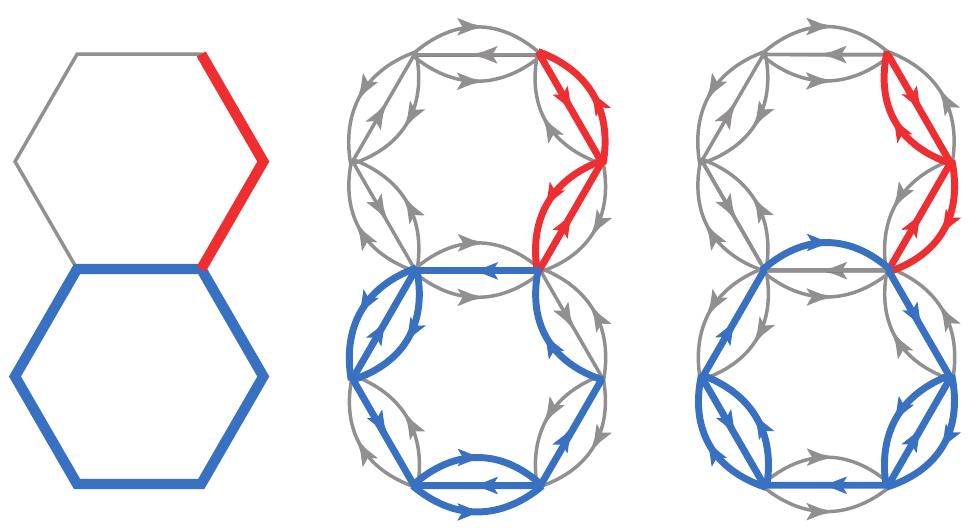}  
		\end{center}
		\caption{A double random current configuration and two corresponding alternating flows with opposite orientations of the outer boundary.}
	\label{fig:flowcurrs}
\end{figure}

Define a map $\theta : \flows^B \to \currs^B$ as follows.
For every $F \in \flows^B$ and every $e \in E$, consider the number of corresponding directed edges $\vec e_m$, $\vec e_{s1}$, $\vec e_{s2}$ present in $F$.
Let $\set_{\rm odd}\subset E$ be the set with one or three such present edges, and $\set_{\rm even} \subset E$ the set with exactly two such edges. Then, set $\theta(F)=\omega$.
It follows from the definition of alternating flows that $\set=\set_{\rm odd} \cup \set_{\rm even}$ is a current with sources $B$. 
Denote by $\theta_* \IP_{\aflow}^B$ the pushforward measure on $\sets^B$. The following result was previously obtained in \cite{LisT}. \begin{theorem}[\cite{LisT}] \label{thm:firstmapping}
For any finite simple planar graph $G$, we have that $\theta_* \IP_{\aflow}^B=\IP_{\drcur}^{B}$.
\end{theorem}

\begin{proof}Since the theorem is a special case of~\cite[Thm 4.1]{LisT}, we only outline the proof here for completeness.

Let $|\set|$ be the number of edges in $\set$ and $k(\set)$ be the total number of clusters of the graph $(V,\set)$ (note that isolated vertices count as a cluster).
Using that the number of even subgraphs of the graph $(V,\omega)$ is equal to $2^{|\omega|-|V|+k(\omega)}$, it can be checked that the double random current measure takes the following form (see \cite[Thm 3.2]{LisT} for a detailed proof): 
\begin{align} \label{eq:doublecurrent}
{\IP}^{B}_{ \drcur}(\set)= \frac1{{Z}^{B}_{ \drcur}} 2^{|\set|+k(\set)} \prod_{e\in \set_{\rm odd}} x_e  \prod_{e\in \set_{\rm even}} x^2_e 
   \prod_{e\in E\setminus\set}   (1-x^2_e).
\end{align}
Now, fix $\set$ and observe that the preimage $\theta^{-1}[\set]$ is simple to understand (see Fig.~\ref{fig:flowcurrs}). Once given the orientations of the boundaries of each one of the non-trivial (meaning not reduced to an isolated vertex) clusters in $F$, not much freedom remains for the edges. More precisely, the even edges of $F$ necessarily contain the edge $e_m$, and the second edge is determined by the alternating condition. An odd edge $e$ can be of two types: either $F$ contains only $e_m$, or it is of a second type, where it contains either $e_{s1}$ only, $e_{s2}$ only, or the three edges $e_{s1}$, $e_m$ and $e_{s2}$. Again, which type it is is determined by the alternating condition. 

Observe that the sum over all configurations in $\theta^{-1}(\set)$ with prescribed orientations of the boundaries of the non-trivial clusters is equal to 
\[
2^{|\set|}\prod_{e\in \set_{\rm odd}} \frac{x_e}{1-x_e^2}  \prod_{e\in \set_{\rm even}} \frac{x^2_e}{1-x_e^2}.
\]
Indeed, each even edge contributes the multiplicative weight $x_{e_m}x_{e_{si}}=2\tfrac{x_e^2}{1-x_e^2}$ (with $i$ equal to 1 or 2), 
each odd edge of the first type $x_{e_m}=\tfrac{2x_e}{1-x_e^2}$, each odd edge of the second type $x_{e_{s1}}+x_{e_{s2}}+x_{e_{s1}}x_{e_m}x_{e_{s2}}=\tfrac{2x_e}{1-x_e^2}$ 
(we take into account that there are three possibilities for the alternating flow at this edge). Finally, each edge not in $\omega$ does not contribute any multiplicative weight. 

The result follows from the fact that the outer boundary of each non-trivial cluster can be oriented in two possible ways, hence the weight $2^{k(\omega)-|V^c(F)|}$.
\end{proof}

We now describe a straightforward measure preserving mapping from the dimer model to alternating flows.
To each matching $\match \in \matchs^B$, we associate a flow $\eta(\match)\in \flows^B$ by replacing each long edge in $\match$ by the corresponding directed edge in $\vec G$.
One can see that this always produces an alternating flow. Indeed, assuming otherwise, there would be two consecutive edges in $\eta(\match)$ of the same orientation,
and therefore, the path of odd length connecting them in a cycle would have a dimer cover, which is a contradiction. 
Let $\eta_*\IP_{\dimer}^B$ be the pushforward measure on $\flows^B$ under the map $\eta$.

\begin{theorem} \label{thm:secondmapping}
For any finite simple planar graph $G$,  we have that $\eta_* \IP_{\dimer}^B=\IP_{\aflow}^{B}$.
\end{theorem}
\begin{proof}
Comparing \eqref{eq:altflow} with \eqref{eq:dimer}, and knowing that the long edges of $G^d$ have the same weights as in $\vec G$, we only need to account for the factor 
$2^{|V^c(F)|}$ from the definition of the alternating flow measure.
To this end, note that the only freedom in the dimer covers in $\eta^{-1}(F)$ is the way they match the short edges in the cycles corresponding to the isolated vertices of $(V,F)$.
Each such cycle has two matchings, and the matchings of different cycles are independent. This completes the proof.
\end{proof}

\begin{proof}[Proof of Theorems~\ref{thm:dimercurr} and \ref{thm:dimercurr2} (i)]
We define $\pi =\theta \circ \eta : \matchs^B \to \currs^B$ to be the many-to-one map projecting dimer covers to currents (note that it is the mapping defined in the introduction). 
Let $\pi_*\IP_{\dimer}^B$ be the pushforward measure on $\currs^B$.
Combining the two previous theorems yields the corresponding statements of the introduction.
\end{proof}

We now turn to height functions. 
Let $h=h_F$ be the {\em height function} of a flow $F$ defined on the faces of $\vec G$ (or $\vec G_{(a,b)}$ if we consider $(a,b)$-boundary conditions) in the following way:
\begin{enumerate}[(i)]
\item $h(u_0)=0$ for the unbounded face $u_0$,
\item for every other face $u$, choose a path $\gamma$ connecting $u_0$ and $u$, and define $h(u)$ to be total flux of $F$ through $\gamma$, i.e., the number of edges in $F$ crossing $\gamma$ from left to right
minus the number of edges crossing $\gamma$ from right to left.
\end{enumerate}
The obtained value is independent of the choice of $\gamma$, since at each $v\in V$, the same number of edges of $h$ enters and leaves $v$ (and so the total flux of $F$ through any
closed path of faces is zero).

\begin{proof}[Proof of Theorem~\ref{thm:nestingfield} and \ref{thm:dimercurr2}(ii)] It is clear that $h_F$ is equal to the height function of the dimer cover $M=\eta(F)$. We therefore relate $h_F$ to $\mathcal S(\omega)$, where $\omega=\theta(F)$. 

Recall from the proof of Theorem~\ref{thm:firstmapping} that for each cluster of a double random current, there are two opposite orientations of the boundary of the corresponding connected component  
of the associated alternating flows in $\theta^{-1}(\omega)$. Set $\xi_{\mathcal C}(F)=+1$ if $F$ is oriented counterclockwise around the boundary of the cluster $\mathcal C$ of $\omega$, and $\xi_{\mathcal C}(F)=-1$ otherwise. By the proof of Theorem~\ref{thm:firstmapping}, $(\omega,\xi(F))$ is in direct correspondence with $F$. Furthermore, by construction,  $h_F$ is equal to the nesting field $\mathcal S(\omega)$ obtained from the $\xi(F)$. The fact that for each cluster $\mathcal C$, the two opposite orientations 
carry the same weight implies that under the law of alternating flows, conditionally on $\omega$, $\xi(F)$ is a iid family of random variables which are equal to $+1$ or $-1$ with probability $1/2$. This concludes the proof.
\end{proof}

\section{Proof of Theorem~\ref{thm:bozonization}}

The proof is based on classical properties of the double random current model combined with the properties of the mapping $\pi$.
First, observing that changing $J_e$ to $-J_e$ amounts to changing $x_e$ to $-x_e$, and not changing $p_e$, 
the classical representation of spin-spin correlations in terms of the random current gives
$$\mu_{G,\beta}^{\rm f}\Big[\prod_{i=1}^k\sigma_{x_i}\times \prod_{j=1}^n\mu_{\ell_j}\Big]=\frac{1}{Z^\emptyset_{\rm curr}}\sum_{\omega\in \Omega^{X}}w(\omega)(-1)^{|\omega_{\rm odd} \cap E_{\text{odd}}|},$$
where $X=\{x_1,\dots,x_k\}$, $E_{\text{odd}}$ is the set of edges crossed an odd number of times by $\cup_{j=1}^n\ell_j$, and where $w(\omega)$ is the weight associated with a current $\omega$ through~\eqref{eq:rc}. 
The switching lemma for double currents \cite{Aiz82,Dum16,Dum17a} implies that 
$$\mu_{G,\beta}^{\rm f}\Big[\prod_{i=1}^k\sigma_{x_i}\times \prod_{j=1}^n\mu_{\ell_j}\Big]^2=\IE_{\drcur}^{\emptyset}[(-1)^{|\omega_{\rm odd}\cap E_{\text{odd}}|}\mathbb  I_{\omega\in \mathcal F_X}],$$
where $\mathcal F_X$ is the event that every cluster of $\omega$ intersects $X$ an even number of times (the points in $X$ are in general counted with multiplicity).
To conclude the proof of the theorem, we therefore need to show that 
\begin{equation}\label{eq:o1}
\IE_{\drcur}^{\emptyset}[(-1)^{|\omega_{\rm odd}\cap E_{\text{odd}}|}\mathbb  I_{\omega\in \mathcal F_X}]=\IE_{\dimer}^\emptyset\Big[\prod_{i=1}^k\sin(\pi h_{x_i})\times\prod_{j=1}^n\cos(\pi h_{u_j})\Big].
\end{equation}
In order to see this, fix a current $\omega$ and denote its nesting field by $\mathcal S$. Observe first that for every dimer configuration $M\in\pi^{-1}(\omega)$, $h_u=\mathcal S_u$ on every face $u$ of the graph and therefore, since all disorder lines start on the unbounded face,
$$\prod_{j=1}^n\cos(\pi h_{u_j})=\prod_{j=1}^n\cos(\pi \mathcal S_{u_j})=(-1)^{|\omega_{\rm odd} \cap E_{\text{odd}}|}.$$
(In particular it does not depend on $M$.)

Let us now treat the case of the product of sines in \eqref{eq:o1}. The definition of the reference 1-form~$f_0$ together with the structure of the graph $G^d$ imply that $h$ is constant on the vertices of every cluster of $\omega$. In particular, if $|X\cap\mathcal C|$ is even for every cluster $\mathcal C$ of $\omega$, the product of sines is equal to 1. To treat the case where $|X\cap\mathcal C|$  is odd for some $\mathcal C$, observe that while the height function of $M$ is not determined by $\omega=\pi(M)$, it is determined by the alternating flow $F=\eta(M)$, except on isolated vertices, where it is obtained by adding $\pm\tfrac12$ to the height function at neighboring faces, independently for each isolated vertex. Since the orientations of the clusters of $F$ are chosen uniformly at random in the coupling introduced in the previous section (they are given by the $\xi_{\mathcal C}$), we conclude that 
$$\IE_{\dimer}^\emptyset\Big[\,\prod_{i=1}^k\sin(\pi h_{x_i})\,\Big|\,\pi(M)=\omega\,\Big]=\mathbb  I_{\omega\in\mathcal F_X}.$$
By combining the two displayed equations with Theorems~\ref{thm:dimercurr} and~\ref{thm:nestingfield}, we deduce \eqref{eq:o1}.

\begin{remark}\label{rmk:1} The relations obtained in Theorem~\ref{thm:bozonization} are the same as in Lemma~3 of \cite{Dub}. Indeed, 
the dimer model on $G^d$ is associated with the dimer model of \cite{Dub,BoudeT} as follows. Given an edge of $G$, 
select a quadrilateral face in $G^d$ corresponding to the edge and (if necessary) split each vertex that the chosen quadrilateral shares
with a quadrilateral corresponding to a different edge of $G$. In this way we find ourselves in the situation from the upper left panel in Fig.~\ref{fig:urbanrenewal}.
After performing urban renewal (i.e.~the transformation from Fig.~\ref{fig:urbanrenewal}) 
and collapsing the doubled edge, we are left with one quadrilateral as desired. One can check that the weights that we obtain
match those from Fig.~\ref{fig:graphs}. We then repeat the procedure for every edge of $G$ and the resulting graph is $C_G$.

Note that the height function on faces is not modified by vertex splitting and urban renewal. Nonetheless, there is indeed loss of information between the dimer model on $G^d$ and the one on $C_G$, and we the former is more suitable for understanding double random currents.
\end{remark}

\begin{figure}
		\begin{center}
			\includegraphics[scale=1]{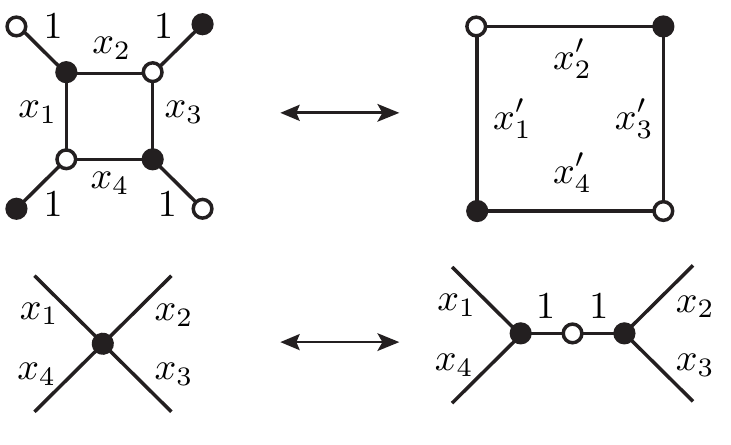}  
		\end{center}
		\caption{Urban renewal and vertex splitting  are transformations of weighted graphs preserving the distribution of dimers and the height function outside the modified region. 
		The weights in urban renewal satisfy $x_1'=\tfrac{x_3}{x_1x_3+x_2x_4}$, 
		$x_2'=\tfrac{x_4}{x_1x_3+x_2x_4}$, $x_3'=\tfrac{x_1}{x_1x_3+x_2x_4}$, $x_4'=\tfrac{x_2}{x_1x_3+x_2x_4}$. }
	\label{fig:urbanrenewal}
\end{figure}

\section{Proof of Theorem~\ref{thm:continuity}}

We will in fact work with the Ising model on the dual graph $\mathbb G^*$ obtained by putting a vertex in each face of $\mathbb G$, and edges between vertices corresponding to neighboring faces. As such, the Ising model below will be seen as a random assignment of spins to the faces of $\mathbb G$. While we use the notation $\mathbb G$ as in the introduction, the outcome of the proof will be Theorem~\ref{thm:continuity} for $\mathbb G^*$. Since the dual graph of a nondegenerated biperiodic graph is itself non-degenerated and biperiodic, this is sufficient. The reason for working with the Ising model on $\mathbb G^*$ is that we will use the connection with the dimer on $\mathbb G$, and that this makes the study more coherent with other sections of the article.

Below, we restrict our attention to the Ising model on $\mathbb G^*$ at $\beta=\beta_c(\mathbb G^*)$ and drop $\beta$ from the notation. Let $\mu^{+}_{\mathbb G^*}$ (resp.\ $\mu^{-}_{\mathbb G^*}$) be the infinite volume Ising measure on $\mathbb G^*$ with~$+$ (resp.\ $-$) boundary conditions, and 
for a face $u$ of $\mathbb G$, let $\mathbf C_u(\sigma)$ be the minimum number of spin changes in $\sigma$ over infinite self-avoiding paths starting from $u$. The architecture of the proof is the following:
\begin{description}
\item[Step 0] We introduce relevant auxiliary infinite volume measures.
\item[Step 1] We show that $\mu^+_{\mathbb G^*}[\mathbf C_u(\sigma)]=+\infty$.
\item[Step 2] We prove that $\mu^+_{\mathbb G^*}[\mathbf C_u(\sigma)=0]=0$.
\item[Step 3] We deduce that $\mu^+_{\mathbb G^*}[\sigma_u]=0$.
\end{description}
\begin{remark}Note that Step 2 can be restated as follows: there is no infinite cluster of pluses or minuses $\mu^+_{\mathbb G^*}$-almost surely. As a byproduct, we obtain Corollary~\ref{cor:5}.
\end{remark} 
Step 1 is the major novelty of the proof. It relies on Theorem~\ref{thm:nestingfield} and Corollary~\ref{cor:4}. Step 3 is directly extracted from \cite[Prop.~4.1]{Dum17a}. We refer to \cite{FriVel17} for classical facts on the Ising model.
\medbreak
Let $\Lambda\approx \mathbb Z\oplus\mathbb Z$ be a group acting transitively on $\mathbb G$. Let $\mathbb G_n=\mathbb G/(n\mathbb Z\oplus n\mathbb Z)$ be the toroidal graph 
of size $n\in\mathbb N$, and let $\mathbb G^d_n = \mathbb G^d/(n\mathbb Z\oplus n\mathbb Z)$ be the bipartite toroidal dimer graph corresponding to $\mathbb G_n$. 
 Below, we consider the random current, double random current and dimer models on $\mathbb G_n$ and $\mathbb G^d_n$ with $n$ tending to infinity, and where the weights $x_e$ on 
$\mathbb G_n$ are defined as follows: if $e$ is the edge between the faces $u$ and $v$, then 
$x_e:=\exp[-2\beta(\mathbb G^*) J_{\{u,v\}}]$. In what follows, we add subscripts to the already introduced notation to 
mark the dependency of the probability measures on the underlying graph.
\bigbreak\noindent{\bf Step 0.} 
Note that for topological reasons, some current configurations on $\mathbb G_n$ do not correspond to spin configurations on the faces of $\mathbb G_n$. To overcome this obstacle, we will resort to the construction 
of infinite volume measures for the different models, where planarity is recovered in the limit as $n$ tends to infinity. There are several ways to proceed and we simply explain here the shortest one (this is not the most self-contained one).

By \cite{KenOkoShe06}, $\IP^{\emptyset}_{\dimer,\mathbb G^d_n}$ converges weakly to a $\Lambda$-invariant measure 
$\IP^{\emptyset}_{\dimer,\mathbb G^d}$ on dimer covers of $\mathbb G^d$. 
Since the sourceless double random current on $\mathbb G_n$ is a local function of the dimer model on $\mathbb G^d_n$, we get that $\IP^{\emptyset}_{\drcur,\mathbb G_n}$ converges
weakly to an infinite volume measure  $\IP^{\emptyset}_{\drcur,\mathbb G}$ on sourceless currents on $\mathbb G$. 

The measures $\IP^{\emptyset}_{\rcur,\mathbb G_n}$ also converge weakly to a measure
$\IP^{\emptyset}_{\rcur,\mathbb G}$ on sourceless currents on~$\mathbb G$. In order to see this, we go back to the original definition of single and double currents in terms of integer-valued functions. Since the integer value of the double current at an edge is obtained from the parity independently for any edge, the integer-valued double random current also converges. With this definition, the integer-valued double random current is simply the sum of two iid integer-valued single random currents, and therefore for any finite set $D$ of edges, the characteristic function of the latter when restricted to $D$ is the square-root of the characteristic function of the former. In particular, it converges for any fixed $D$. This implies the convergence of the single random current.

We now define a probability measure $\mu_{\mathbb G^*}$ on the space of $\pm 1$ spin configurations on the faces of $\mathbb G$ by tossing a symmetric coin to decide the spin at a fixed face, and then 
using the odd part of a current $\omega$ drawn according to $\IP^{\emptyset}_{\rcur,\mathbb G}$ to define the interfaces between $+1$ and $-1$ spins. 
This is well defined since $\mathbb G$ is planar and the degree of $\omega_{\rm odd}$
at every vertex of $\mathbb G$ is even almost surely. Note that $\mu_{\mathbb G^*}$ is $\Lambda$-invariant since the infinite-volume version of the single random currents inherits the invariance under the action of $\Lambda$ from the dimer measure. 

Using
the domain Markov property of $\omega_{\rm odd}$ under $\IP^{\emptyset}_{\rcur,\mathbb G_n}$, and the
fact that a spin configuration under $\mu_{\mathbb G^*}$ carries the same information (up to a spin flip) as $\omega_{\rm odd}$, 
one can check that $\mu_{\mathbb G^*}$ satisfies the Dobrushin--Lanford--Ruelle conditions for an infinite volume
Gibbs state of the Ising model with parameters $\beta$ and $(J_{e})_{e\in E}$. 

A result of Raoufi~\cite{Rao17} classifying $\Lambda$-invariant Gibbs measures for the Ising model, and the $\pm1$ symmetry of $\mu_{\mathbb G^*}$ readily yield 
\begin{align} \label{eq:concom}
\mu_{\mathbb G^*} = \tfrac12( \mu_{\mathbb G^*}^+ +\mu_{\mathbb G^*}^-).
\end{align} 
(Note that the result in~\cite{Rao17} is stated for vertex transitive graphs, and it can be generalized to the quasi-transitive case which includes biperiodic graphs).

\bigbreak\noindent{\bf Step 1.} 
Fix two faces $u$ and $v$ of $\mathbb G_n$, and a self-avoiding path $\gamma$ connecting $u$ and $v$. 
Recall that a cluster $\clust$ of $\omega$ is {odd with respect to} $\gamma$ if $\clust \cap \omega_{\rm odd}$ crosses an odd number of edges of $\gamma$.
For $k=1,2,\ldots,\infty$, let $\mathbf N^k_{\gamma}(\set)$ be the number of 
clusters of $\omega$ that are odd with respect to $\gamma$ in the current configuration obtained by restricting $\omega$ to the set of edges at distance at most $k$ from~$\gamma$.
The quantities $\mathbf N^k_{\gamma}$ are subadditive, i.e.
\begin{align} \label{eq:subadd}
\mathbf N^k_{\gamma}(\set+\set')\leq \mathbf N^k_{\gamma}(\set) +\mathbf N^k_{\gamma}(\set').
\end{align}
We give a proof of this inequality that is independent of $k$ so we may assume that $k=\infty$. Indeed, note that if $\clust$ is a cluster of $\set+\set'$,
then the parity of the number of edges in $\clust \cap (\set+\set')_{\rm odd}$ crossing $\gamma$ is equal to the parity of the sum of the numbers of edges in 
$\clust \cap \set_{\rm odd}$ and $\clust \cap \set'_{\rm odd}$ crossing $\gamma$. Hence, if the former number is odd, exactly one of the latter numbers is odd, which means
that either $\omega$ or $\omega'$ contain at least one cluster that is odd with respect to $\gamma$, and \eqref{eq:subadd} is proved.

Note moreover that $\mathbf  N^k_{\gamma}$ is decreasing in $k$ since adding connections to the current cannot result in a larger number of odd clusters. By \eqref{eq:subadd}
and Remark~\ref{rem:torusfield}, we can therefore write for all $k$ and $n$,
\begin{align} \label{eq:bound1}
\IE^{\emptyset}_{\rcur,\mathbb G_n} [ N_{\gamma}^k] \geq \tfrac12\IE^{\emptyset}_{\drcur,\mathbb G_n}  [ N_{\gamma}^k] \geq \tfrac12\IE^{\emptyset}_{\drcur,\mathbb G_n}  [ N_{\gamma}^{\infty}]
=\tfrac12\IE^{\emptyset}_{\drcur,\mathbb G_n}  [ \mathcal{S}_{\gamma}^2] = \tfrac12\IE^{\emptyset}_{\dimer,\mathbb G^d_n}  [ h_{\gamma}^2],
\end{align}
where $\mathcal S_{\gamma}$ and $h_{\gamma}$ are the increments along $\gamma$ of the nesting field and the dimer height function respectively. 
The first equality follows from the fact that conditionally on $\omega$, $\field_{\gamma}$ is the sum of $\mathbf N_{\gamma}(\omega)$ iid centered random variables of variance 1.
As both $N_{\gamma}^k$ and $h_{\gamma}$ are local functions, taking first the weak limit in $n$ and then the decreasing limit in $k$ on both sides of \eqref{eq:bound1}, we get
\begin{align}\label{eq:bound2}
\IE^{\emptyset}_{\rcur,\mathbb G} [ N_{\gamma}^{\infty}] \geq \tfrac12\IE^{\emptyset}_{\dimer,\mathbb G^d}  [ h_{\gamma}^2]=\tfrac{1}{2\pi}\log[|\phi(u)-\phi(v)|] +o(\log[|\phi(u)-\phi(v)|]),
\end{align}
where in the equality, we used Corollary~\ref{cor:4} together with the fact that at $\beta=\beta_c$, the relation between dimers on $C_{\mathbb G}$ and $\mathbb G^d$ implies by \cite{CimDum13} that the characteristic polynomial of the dimer model on $\mathbb G^d$ has a real zero on the torus $\mathbb T^2$. (Recall that $\phi$ is a linear transformation.)

For a spin configuration $\sigma$ on the faces of $\mathbb G$, let $\mathbf C_{uv}(\sigma)$ be the minimal number of sign changes in $\sigma$ along all self-avoiding paths from $u$ to $v$.
It follows that every such path $\gamma$ should contain at least one spin change per odd cluster, and therefore $\mathbf C_{uv}(\sigma)\ge \mathbf N^{\infty}_{\gamma}(\set)$, where $\sigma$ and $\omega$ are related by the low-temperature expansion (hence the choice of $x_e$ at the beginning of the proof). We deduce that 
\begin{equation}
\label{eq:kk}\mu_{\mathbb G^*}[\mathbf C_{uv}(\sigma)]\ge \IE^{\emptyset}_{\rcur,\mathbb G} [\mathbf N^{\infty}_{\gamma}(\omega)],
\end{equation}
and together with \eqref{eq:bound2} this gives us
\[\mu_{\mathbb G^*}[\mathbf C_u(\sigma)]+\mu_{\mathbb G^*}[\mathbf C_v(\sigma)]\ge \mu_{\mathbb G^*}[\mathbf C_{uv}(\sigma)]\ge\tfrac{1}{2\pi}\log[|\phi(u)-\phi(v)|] +o(\log[|\phi(u)-\phi(v)|]).
\]
Letting $|u-v|$ tend to infinity, and using the invariance of $\mu_{\mathbb G^*}$ under the action of $\Lambda$, we find that 
$\mu_{\mathbb G^*}[\mathbf C_u(\sigma)]=+\infty$ for every face $u$.
To complete this step, it only remains to transfer this estimate to $\mu_{\mathbb G^*}^+$ instead of $\mu_{\mathbb G^*}$.
But since $\mathbf C_u(\sigma)=\mathbf C_u(-\sigma)$, by \eqref{eq:concom} 
we deduce that $\mu_{\mathbb G^*}^+[\mathbf C_u(\sigma)]=\mu_{\mathbb G^*}^-[\mathbf C_u(\sigma)]=\mu_{\mathbb G^*}[\mathbf C_u(\sigma)]$.
\bigbreak
Below, we will use the following notation. For a set of faces $F$, $\partial F$ denotes the set of faces $u$ such that there exists a neighboring face $v$ which is not in $F$.
\bigbreak\noindent{\bf Step 2.} 
We proceed by contradiction. Assume that
$\mu^+_{\mathbb G^*}[\mathbf C_u(\sigma)=0]=p>0.$
We wish to prove that for every $k\ge0$,
\begin{equation}\label{eq:a}\mu^+_{\mathbb G^*}[\mathbf C_u(\sigma)\ge k+2\ |\ \mathbf C_u(\sigma)\ge k]\le 1-p.\end{equation}
This immediately implies that $\mu^+_{\mathbb G^*}[\mathbf C_u(\sigma)]<\infty$, which contradicts the first step.

Note that it suffices to show that for every $k\ge0$,
\begin{align}\label{eq:a}&\mu^+_{\mathbb G^*}[\mathbf C_u(\sigma)\ge 2k+1\ |\ \mathbf C_u(\sigma)\ge 2k\text{ and }\sigma_u=+]\le 1-p,\\
&\mu^+_{\mathbb G^*}[\mathbf C_u(\sigma)\geq 2k+2\ |\ \mathbf C_u(\sigma)\ge2k+1\text{ and }\sigma_u=-]\le 1-p.\nonumber\end{align}
We prove the first inequality, the second follows similarly. 

For $k=0$, the result is a direct consequence of the definition of $p$. For $k\ge1$, let $\mathcal F$ be the set of faces $v$ of $\mathbb G$ for which every path from $u$ to $v$ contains at least $2k$ changes of signs. Fix a set of faces $F$. For $\sigma\in A_F:=\Sigma_{\mathbb G^*}\cap\{\mathcal F=F\}\cap\{\mathbf C_u(\sigma)\ge 2k\}\cap\{\sigma_u=+1\}$, faces on $\partial F$ have spin $+1$. Therefore, we deduce that 
\begin{align*}\mu^+_{\mathbb G^*}[\mathbf C_u(\sigma)\ge 2k+1|A_F]& \leq 1-\mu^+_{\mathbb G^*}[S_F\,|\,\sigma_v=+1\text{ for all }v\in \partial F],\end{align*}
where $S_F$ denotes the event that there is an infinite self-avoiding path of pluses starting from $\partial F$. Note that $\mathbf C_u(\sigma)=0$ is included in $S_F$. The FKG inequality for the Ising model implies immediately that 
$$\mu^+_{\mathbb G^*}[S_F|\sigma_v=+1\text{ for all } v\in \partial F]\ge \mu^+_{\mathbb G^*}[S_F]\ge \mu^+_{\mathbb G^*}[\mathbf C_u(\sigma)=0]=p,$$
so that
\begin{align*}\mu^+_{\mathbb G^*}[\mathbf C_u(\sigma)\ge 2k+1|A_F]& \le 1-p.\end{align*}
Since the events $A_F$ partition $\Sigma_{\mathbb G^*}\cap\{\mathbf C_u(\sigma)\ge 2k\}\cap\{\sigma_u=+1\}$, summing on all possible $F$ gives \eqref{eq:a}.

\bigbreak\noindent{\bf Step 3.} It suffices to show that $\mu^+_{\mathbb G^*}[\sigma_u]\le 0$ since we already know by the first Griffiths inequality that $\mu^+_{\mathbb G^*}[\sigma_u]\ge0$. 

Fix a finite subgraph $H$ of $\mathbb G^*$ and note that 
\begin{equation}\label{eq:h}\mu^+_H[\sigma_u]\le \mu^+_H[\mathbf C_u(\sigma)=0]+\mu^+_H[\sigma_u\mathbf{1}_{\mathbf C_u(\sigma)\ge1}].\end{equation}
Now, condition on the set $\mathcal F$ of faces of $\mathbb G$ which are not connected by a path of pluses to the exterior of $H$.  By definition, conditioned on $\mathcal F=F$, the configuration outside $F$ is made of pluses, and the configuration inside of $F$ is an Ising model conditioned on faces of $\partial F$ to have spin $-1$. Furthermore,  $\mathbf C_u(\sigma)\ge1$ implies that $u\in F$. Thus, the Gibbs property implies that \begin{align}\label{eq:hh}\mu^+_H[\sigma_u\mathbf{1}_{\mathbf C_u(\sigma)\ge1}]&=\sum_{F\ni u}\mu^+_H[\sigma_u|\sigma_v=-1,\forall v\in \partial F]\times\mu^+_H[\mathcal F=F\text{ and }\mathbf C_u(\sigma)\ge1]\le0,\end{align}
where the inequality follows from the fact that the FKG inequality implies that 
$$\mu^+_H[\sigma_u|\sigma_v=-1,\forall v\in \partial H]\le 0.$$
Plugging this in \eqref{eq:h} gives that $\mu^+_H[\sigma_u]\le \mu^+_H[\mathbf C_u(\sigma)=0]$. Step 2 implies that  $\mu^+_{\mathbb G^*}[\sigma_u]\le 0$ by letting $H$ tend to $\mathbb G^*$.

\bibliographystyle{amsplain}

\end{document}